\documentclass[%
 reprint,
superscriptaddress,
nofootinbib,
 amsmath,amssymb,
 aps,
prb,
showkeys
]{revtex4-2}
\usepackage{graphicx}  
\usepackage{dcolumn}   
\usepackage{bm}        
\usepackage{amssymb}   
\usepackage{siunitx}
\usepackage{hyperref}
\usepackage{multirow}
\usepackage{xcolor}
\usepackage{amsmath}
\usepackage{hhline}
\usepackage{lipsum}

\usepackage{soul}
\usepackage{color} 

\usepackage{braket}
\usepackage{autobreak}

\usepackage{amsfonts}
\usepackage{amsmath}
\usepackage{slashed}
\usepackage{amsthm}
\usepackage{amssymb}

\newtheorem{theorem}{Theorem}

\begin{document}


\title{Time-Reversal Soliton Pairs in Even Spin-Chern-Number Higher-Order Topological Insulators}


\author{Yi-Chun Hung}
\affiliation{Department of Physics, Northeastern University, Boston, Massachusetts 02115, USA}
\affiliation{Institute of Physics, Academia Sinica, Taipei 115201, Taiwan}
\affiliation{Department of Physics, National Taiwan University, Taipei 106216, Taiwan}

\author{Baokai Wang}
\affiliation{Department of Physics, Northeastern University, Boston, Massachusetts 02115, USA}

\author{Chen-Hsuan Hsu}
\affiliation{Institute of Physics, Academia Sinica, Taipei 115201, Taiwan}
                
\author{Arun Bansil}
\affiliation{Department of Physics, Northeastern University, Boston, Massachusetts 02115, USA}

\author{Hsin Lin}
\email{nilnish@gmail.com}
\affiliation{Institute of Physics, Academia Sinica, Taipei 115201, Taiwan}


\begin{abstract}
{Solitons formed through the one-dimensional mass-kink mechanism on the edges of two-dimensional systems with non-trivial topology play an important role in the emergence of higher-order (HO) topological phases. In this connection, the existing work in time-reversal symmetric systems has focused on gapping the edge Dirac cones in the presence of particle-hole symmetry, which is not suited to the common spin-Chern insulators. Here, we address the emergence of edge solitons in spin-Chern number of $2$ insulators, in which the edge Dirac cones are gapped by perturbations preserving time-reversal symmetry but breaking spin-$U(1)$ symmetry. Through the mass-kink mechanism, we thus explain the appearance of pairwise corner modes and predict the emergence of extra charges around the corners. By tracing the evolution of the mass term along the edge, we demonstrate that the in-gap corner modes and the associated extra charges can be generated through the $S_z$-mixing spin-orbit coupling via the mass-kink mechanism. We thus provide strong evidence that an even spin-Chern-number insulator is an HO topological insulator with protected corner charges.}
\end{abstract}

\maketitle

\section{Introduction}\label{section:I}
Analysis of topological properties of the ground state electronic structures has yielded many new insights into the nature of quantum matter. The bulk-boundary correspondence not only manifests itself through the appearance of gapless edge states, but it can also lead to topological solitons and the associated topologically protected charges \cite{PhysRevD.13.3398, PhysRevB.22.2099, RevModPhys.60.781, PhysRevB.100.205126, PhysRevLett.46.738, PhysRevLett.49.1455}. These solitons are typically generated at the interfaces that separate two regions characterized by (complex) mass terms with different phases, which we will refer to as \emph{the mass-kink mechanism} in one-dimensional (1-D) systems. In the low-energy states on the edges of a two-dimensional (2-D) topological insulator, different phases of the massive edge-Dirac-cones can also induce corner solitons, indicating the presence of higher-order (HO) topological phases with protected extra corner charges \cite{PhysRevB.99.155102, PhysRevB.97.205134, PhysRevB.100.205126, Mu2022, PhysRevResearch.2.033071, PhysRevLett.126.066401}. In this connection, the existing literature in time-reversal-symmetric systems has focused on gapping the Dirac cones with particle-hole symmetry \cite{PhysRevB.77.041406, PhysRevB.92.235435, PhysRevLett.108.196804, PhysRevLett.121.196801, PhysRevLett.111.056403, PhysRevB.99.075129}. However, the role of the mass-kink mechanism in gapping Dirac cones in particle-hole symmetry-breaking systems is not well understood. Here, we examine how, in even spin-Chern-number insulators, the edge states can be gapped through spin-$U(1)$-symmetry-breaking terms in the Hamiltonian. 

Even spin-Chern-number insulators have been demonstrated in several recent tight-binding models \cite{PhysRevB.94.235111, PhysRevB.104.L201110}. In the absence of $S_z$-mixing spin-orbit coupling (SOC), these models support insulators with a spin-Chern number ($C_s$) of 2, which implies the presence of two spin-polarized gapless edge states protected by $S_z$ conservation (Fig.~\ref{fig:01}). When the $S_z$-mixing SOC is turned on, the edge states are gapped (Fig.~\ref{fig:02}) without changing the spin-Chern number \cite{PhysRevB.80.125327, shulman2010robust}. The corner modes appear pairwise within the edge band gap in a finite-sized thin film due to time-reversal symmetry. However, the origin of these corner modes is not well understood. 

Since the edge states are gapped, the non-vanishing mass terms are driven by the $S_z$-mixing SOC. The mass kink could thus be anticipated to account for the corner modes. Through Bosonization, we will show how the corner modes and the associated corner charges can be attributed to the two time-reversal-related edge solitons via the mass-kink mechanism. These solitons form a \emph{time-reversal-related soliton pair} \cite{Schindler2022} and provide an explanation of the existence of pairwise in-gap corner modes with extra charges in a finite-sized thin film.
We illustrate our new approach by considering the model of Ref. \cite{PhysRevB.104.L201110} as a concrete example with $C_s=2$ and $S_z$-mixing SOC-induced mass terms.
By tracking the phase evolution of the corner modes along the edge of a thin film, we show that the corner modes and the associated extra charge arise from a time-reversal soliton pair created by the $S_z$-mixing SOC-induced mass-kink mechanism. These results indicate that even in the absence of gapless edge Dirac cones derived from crystalline symmetries, an HO topological insulator with protected corner charges can emerge from a $C_s=2$ spin-Chern insulator.

This article is organized as follows. In Section~\ref{section:II}, we construct an effective model for the edge-hosting two Dirac cones, along with a general mass term. We show how a time-reversal-related soliton pair arises in the system, leading to an HO topological phase with corner modes. Through a renormalization group (RG) analysis, we show that our results are robust against the presence of weak electron-electron interaction. In Section~\ref{section:III}, we numerically confirm our results by considering a concrete example of a time-reversal-symmetric generalization of the Haldane model in which the mass term is induced via $S_z$-mixing SOC. We summarize our results in Section~\ref{section:IV}. Details of the formalism in the main text are given in the appendices. Bosonization conventions, some basic properties, and the general mass term are presented in Appendix~\ref{section:A}, Appendix~\ref{section:B}, and Appendix~\ref{section:C}, respectively.

\begin{figure}[ht]
  \centering
  \centering
    \includegraphics[width=8.6cm, height=7.5cm]{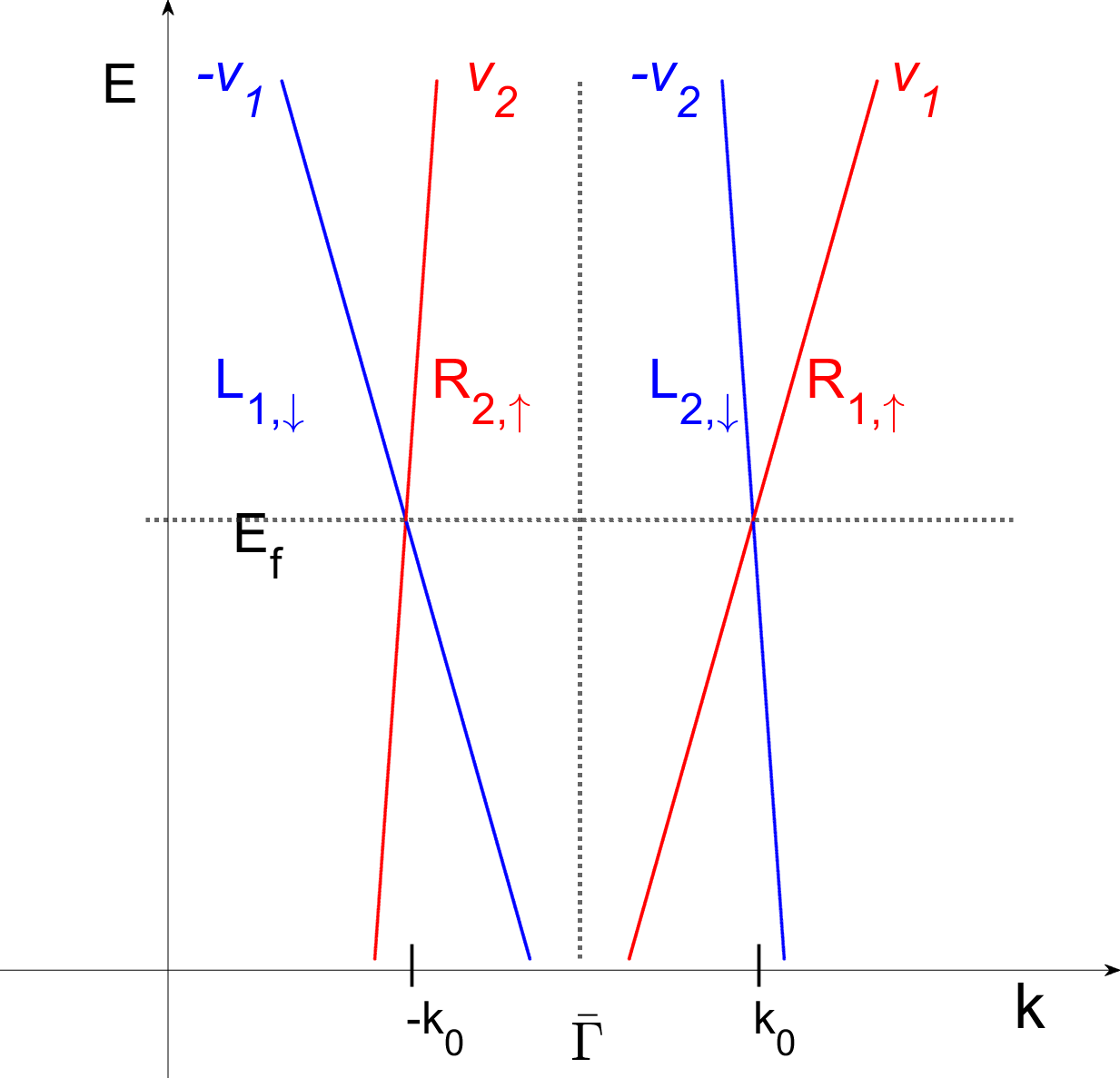}
  \caption{A schematic diagram of our picture of the edge modes before adding the gap-opening perturbation, where the index $1(2)$ indicates the $1^{st}(2^{nd})$ time-reversal sector; $R(L)$ refers to the right (left ) moving mode and $v$ is the Fermi velocity. $\pm k_0$ values of the edge Dirac cones in the edge Brillouin zone are labeled.}
  \label{fig:01}
\end{figure}

\begin{figure}[ht]
  \centering
  \centering
    \includegraphics[width=8.6cm, height=7.5cm]{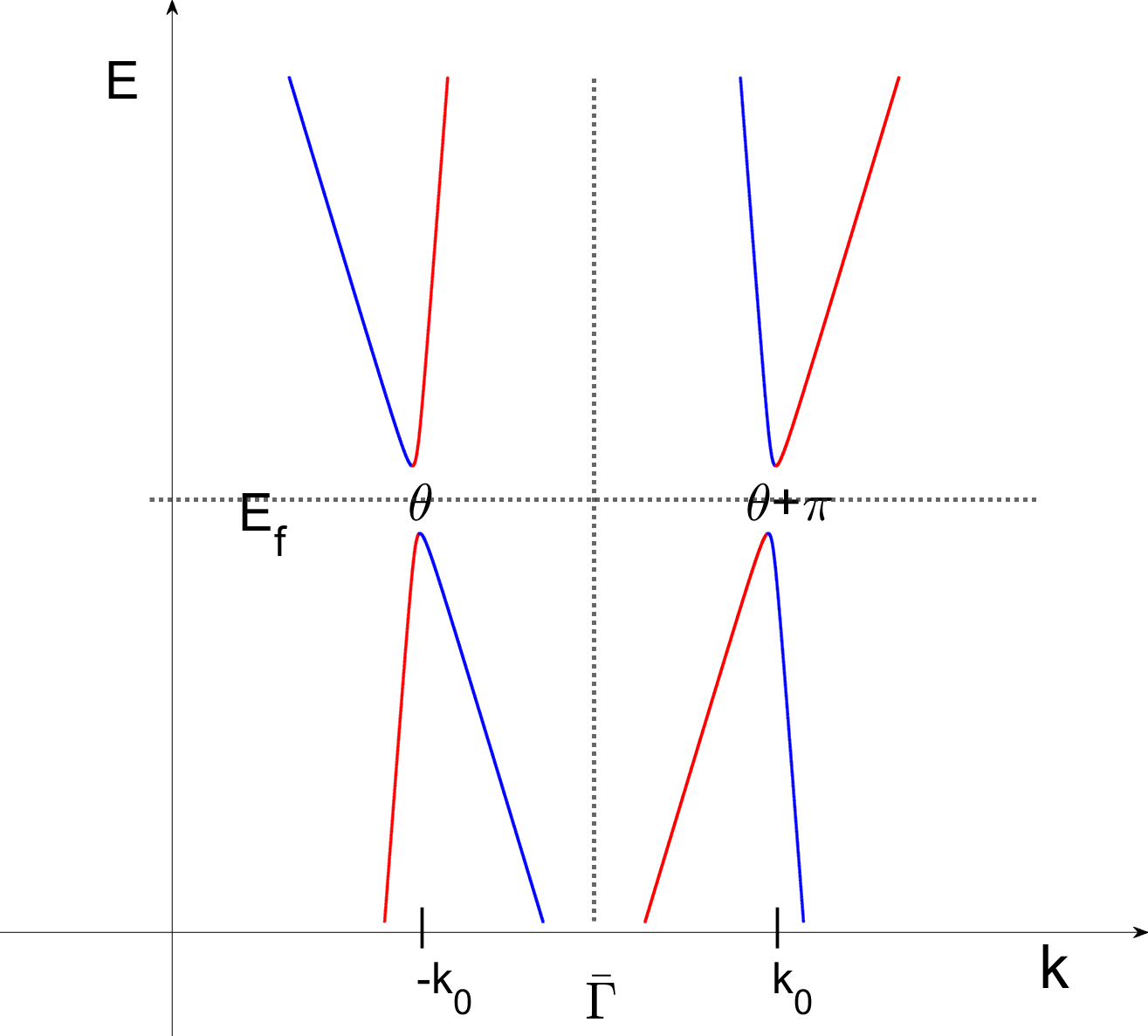}
  \caption{A schematic diagram of our picture of the edge modes after adding the gap-opening perturbation. With the presence of time-reversal symmetry, the mass terms of different Dirac cones have a $\pi$-phase difference in their phase $\theta$. In our case, the perturbation is the $S_z$-mixing SOC.}
  \label{fig:02}
\end{figure}

\section{The Edge Theory of a $C_s=2$ Spin-Chern insulator}\label{section:II}
\par We start by considering an edge theory of spin-Chern insulator with $C_s=2$. Before introducing perturbations, we have two edge Dirac cones located at the momenta $\pm k_0$ (with $k_0>0$) at the Fermi level due to the non-vanishing spin-Chern number, as displayed in Fig.~\ref{fig:01}. Due to the helical nature of the edge states, the right-moving modes are locked with spin up while the left-moving modes lock with spin down (Fig.~\ref{fig:01}). Unlike systems hosting a single Dirac cone at a time-reversal invariant momentum, each of the Dirac cones is not formed by time-reversal partners. Rather, each of them emerges from right-moving and left-moving modes deriving from distinct time-reversal sectors. Therefore, there is no Kramers' degeneracy in our low-energy theory. Since different time-reversal sectors can have different Fermi velocities, the Dirac cones are slightly tilted (Fig.~\ref{fig:01}).

\par Following the strategy used in the literature \cite{PhysRevB.99.155102, PhysRevB.92.035139, PhysRevB.93.235436}, we map our edges around a corner into a 1-D system in which the corner is at \(x=0\) connecting two edges represented by the $x>0$ and $x<0$ regions. We write the edge-theory model shown in Fig.~\ref{fig:01} into real space as:
\begin{equation}\label{eq:07}
\begin{split}
    H_0(x) & = -iv_1(R^\dagger_1 \partial_x R_1 - L^\dagger_1 \partial_x L_1) 
    \\ & \quad -iv_2(R^\dagger_2 \partial_x R_2 - L^\dagger_2 \partial_x L_2),
\end{split}
\end{equation} 
where \(v_{1(2)}\) is the Fermi velocity for the \(1^{st}(2^{nd})\) time-reversal sector. We then define the \emph{average Fermi velocity} \(v=\frac{v_1+v_2}{2}\) and \emph{its deviation between two time-reversal sectors} \(\delta v=\frac{v_2-v_1}{2}\). Eq.(\ref{eq:07}) then becomes:
\begin{equation}\label{eq:08}
\begin{split}
    H_0(x) & = -iv\sum_{\mu=1}^2(R^\dagger_\mu \partial_x R_\mu - L^\dagger_\mu \partial_x L_\mu) 
    \\ & \quad -i\delta v\sum_{\mu=1}^2 \epsilon^{\mu\nu}(R^\dagger_\nu \partial_x R_\nu - L^\dagger_\nu \partial_x L_\nu), 
\end{split}
\end{equation} 
which can be further Bosonized, with the convention in Appendix~\ref{section:A}, into \cite{PhysRevB.92.035139}:
\begin{equation}\label{eq:09}
\begin{split}
    H_0^{(B)}(x) & =
    \frac{1}{2}\sum_{\mu=1}^2 v[(\partial_x\varphi_\mu)^2+(\partial_x\vartheta_\mu)^2] + 
    \\ & \quad + \frac{1}{2}\sum_{\mu=1}^2\delta v\epsilon^{\mu\nu}[(\partial_x\varphi_\nu)^2+(\partial_x\vartheta_\nu)^2] ,
\end{split}
\end{equation}    
which can be further written as:
\begin{equation}\label{eq:10}
\begin{split}
    H_0^{(B)}(x) & = \frac{v}{2}\{[(\partial_x\varphi_c)^2 + (\partial_x\vartheta_c)^2 + (\partial_x\varphi_d)^2+ (\partial_x\vartheta_d)^2] 
    \\ & \quad - 2\frac{\delta v}{v}[\partial_x\varphi_c\partial_x\varphi_d + \partial_x\vartheta_c\partial_x\vartheta_d]\} .
\end{split}
\end{equation} 
Upon adding perturbations that break the spin rotational symmetry but preserve the time-reversal symmetry, the Dirac points can be gapped, and the low-energy theory acquires a mass term, which we introduce next. 

\subsection{Mass Term} 
\par We consider time-reversal-symmetric perturbations that can open band gaps at the Dirac cones described in Eq.(\ref{eq:07}). While the $S_z$-mixing SOC, which is naturally present in real materials, offers an illustrative example of such perturbations, here we introduce a general mass term for generality.
Consider the operation of time-reversal symmetry in this system first:
\begin{align}\label{eq:11}
    \Theta \begin{pmatrix} R_1 \\ L_2 \end{pmatrix} = \begin{pmatrix} L_1^* \\ -R_2^* \end{pmatrix} ,
\end{align}
with the time-reversal operator $\Theta$. This allows us to regard the Hamiltonian in Eq.(\ref{eq:07}) as a sum of two time-reversal related sub-Hamiltonians \(H_0^{(+)}(x)\) and \(H_0^{(-)}(x)=\Theta H_0^{(+)}(x)\Theta^{-1}\):
\begin{equation}\label{eq:12}
\begin{split}
    H_0(x) & = H_0^{(+)}(x) + H_0^{(-)}(x),
    \\ H_0^{(+)}(x) & = -iv_1R^\dagger_1\partial_x R_1 + iv_2L^\dagger_2\partial_xL_2 .
\end{split}
\end{equation}
The $H_0^{(\pm)}(x)$ here can be understood as describing the Dirac fermion around $\pm k_0$. It is more convenient to have the representation of time-reversal symmetry in the \(\begin{pmatrix} R_1 \, L_2 \, R_2 \, L_1\end{pmatrix}^T\) basis:
\begin{align}\label{eq:13}
    \Theta = \begin{pmatrix} 0 & 0&  0&  1 \\ 0&  0&  -1&  0 \\ 0&  1&  0&  0 \\ -1&  0&  0&  0 \end{pmatrix} K,
\end{align}
with the complex conjugate $K$. 

Hence, if the mass term in one of the sub-Hamiltonian, say, \(H_M^{(+)}(x)\), is given by
\begin{equation}\label{eq:14}
    me^{-i\theta}R^\dagger_1L_2 + me^{i\theta}L^\dagger_2R_1,
\end{equation}
with \(m>0\), the mass term in the full Hamiltonian can be obtained through Eq.(\ref{eq:12}) and Eq.(\ref{eq:13}), which gives 
\begin{eqnarray}\label{eq:15}
\begin{split}
    H_M = &  H_M^{(+)}(x) + H_M^{(-)}(x) \\
             = & ( me^{-i\theta}R^\dagger_1L_2 + H.c.) + ( me^{-i(\theta+\pi)}R^\dagger_2L_1 + H.c.).
\end{split}
\end{eqnarray}
There will thus be a $\pi$ phase difference between the mass terms of the two sub-Hamiltonians $H_M^{(+)}(x)$ and $H_M^{(-)}(x)$, as indicated in Fig.~\ref{fig:02}. 

Using Eq.(\ref{eq:01})--(\ref{eq:05}), Eq.(\ref{eq:15}) can be rewritten as:
\begin{equation}\label{eq:17}
H_M^{(B)} = \frac{2m}{\pi a}\sin(\sqrt{2\pi}\vartheta_d)\sin(\theta+\sqrt{2\pi}\varphi_c),
\end{equation}
which has minima at
\begin{align}\label{eq:18}
\vartheta_d=(1\pm\frac{1}{2})\sqrt{\frac{\pi}{2}},\,\varphi_c=-\frac{\theta}{\sqrt{2\pi}}-(1\pm\frac{1}{2})\sqrt{\frac{\pi}{2}} .
\end{align}

\par Notably, the mass terms here are generated by perturbations that do not close and reopen the bulk band gap and only influence the low-energy theory and hardly change the Fermi velocities determined by the overall band structure. Thus, the matrix elements of the \(H_M\) can be calculated by expanding the perturbation term $\Delta H$ in the Hamiltonian on the eigenstates of \(H_0\), which we demonstrate in Section \ref{section:III} through a concrete example.

\subsection{Bulk-Boundary Correspondence}
\par To derive the conserved current and hence the conserved charge in the system shown in Fig.~\ref{fig:02}, we use the following notation:
\begin{equation}\label{eq:19}
\begin{split}
    & \gamma^0 = \begin{pmatrix} 0 & 1 \\ 1 & 0 \end{pmatrix}, \gamma^1 = \begin{pmatrix} 0 & 1 \\ -1 & 0 \end{pmatrix},\gamma^5 = \begin{pmatrix} -1 & 0 \\ 0 & 1 \end{pmatrix},
    \\ & \psi_\mu = \begin{pmatrix} L_\mu \\ R_\mu \end{pmatrix}, \bar{\psi}_\mu = \psi_\mu^\dagger\gamma^0,
    \quad \mu=1,2.
\end{split}
\end{equation} 
This allows us to rewrite the Hamiltonian with gap-opening perturbation \(H=H_0+H_M\) into:
\begin{equation}\label{eq:20}
\begin{split}
    H & = -i\sum_\mu v\bar{\psi}_\mu\gamma^1\partial_x\psi_\mu + \delta v \epsilon^{\mu\nu}\bar{\psi}_\nu\gamma^1\partial_x\psi_\nu
    \\ & \quad + \sum_{\mu\neq\nu}\bar{\psi}_\mu(-m\cdot \cos\theta\gamma^5+(-1)^\mu i\cdot m\cdot \sin\theta\mathbb{I})\psi_\nu,
\end{split}
\end{equation} 
which gives rise to the corresponding Lagrangian: 
\begin{equation}\label{eq:21}
\begin{split}
    L & = i\sum_\mu[v\bar{\psi}_\mu\slashed{\partial}\psi_\mu - \delta v\epsilon^{\mu\nu}\bar{\psi}_\nu\slashed{\partial}\psi_\nu] 
    \\ & \quad -\sum_{\mu\neq\nu}\bar{\psi}_\mu(-m\cdot \cos\theta\gamma^5 + (-1)^\mu i\cdot m\cdot \sin\theta\mathbb{I})\psi_\nu.
\end{split}
\end{equation} 

Thus, the conserved currents are:
\begin{equation}\label{eq:23}
\begin{split}
    j^\sigma & \propto \sum_\mu \bar{\psi}_\mu\gamma^\sigma\psi_\mu - \frac{\delta v}{v}\epsilon^{\mu\nu}\bar{\psi}_\nu\gamma^\sigma\psi_\nu
    \\ & = \sqrt{\frac{2}{\pi}}(\epsilon^{\sigma\eta}\partial_\eta\varphi_c + \frac{\delta v}{v} \epsilon^{\sigma\xi}\partial_\xi\varphi_d).
\end{split}
\end{equation} 

Therefore, the total charge carried by the solitons is:
\begin{equation}\label{eq:24}
    Q_s^0 \propto e\int j^0 dx
     = e \sqrt{\frac{2}{\pi}} \int (\partial_x\varphi_c - \frac{\delta v}{v}\partial_x\varphi_d) dx,
\end{equation}
where $e$ is the electron charge. Through Eq.(\ref{eq:18}), we approximate \(\partial_\mu\varphi_c\cong\frac{-1}{\sqrt{2\pi}}\partial_\mu\theta\) near the minimum of \(H_M^{(B)}\) and apply Eqs.(\ref{eq:04})--(\ref{eq:05}) to Eq.(\ref{eq:24}):
\begin{equation}
\begin{split}\label{eq:25}
    Q_s^0 & \cong -\frac{e}{\pi} \int (\partial_x\theta + \pi\frac{\delta v}{v}(n_1 - n_2)) dx
    \\ & = -\frac{e}{2\pi}2(\theta(x\rightarrow\infty)-\theta(x\rightarrow-\infty)) - e\frac{\delta v}{v}(N_1 - N_2),
\end{split}
\end{equation} 
where \(n_i\) and \(N_i\) are the particle number density and the particle number of the \(i^{th}\) time-reversal sector, respectively. In the absence of external fields that tilt the Fermi level, we have $N_1 = N_2$ and hence
\begin{align}\label{eq:25.5}
    Q_s^0 = -\frac{e}{2\pi}2(\theta(x\rightarrow\infty)-\theta(x\rightarrow-\infty)).
\end{align}
 Note that, since $e$ is negative, the soliton charge $Q_s^0$ is positive, and the soliton particle number \(N_s = \frac{Q_s^0}{e}\) is negative. 
 While we focus on time-reversal-invariant systems here, Eq.~(\ref{eq:25.5}) can apply to systems without time-reversal symmetry, indicating that the soliton pairs can be created by applying time-reversal symmetry-breaking perturbations that gap the edge states protected by some symmetries \cite{PhysRevB.94.235111, PhysRevB.104.L201110} (see Appendix~\ref{section:C}).
 
\par Since the charge density $j^0$ is proportional to $\partial_x\theta$ according to the above approximations, a soliton appears around a spatial point where $\theta$ develops a kink. For smooth and clean edges where local perturbations due to spatial inhomogeneity are absent, we expect such a kink in $\theta$ only around a corner at which the adjacent edges are characterized by different values of $\theta$. Therefore, the bulk-boundary correspondence states that the total soliton charge \(Q_s^0\) around a corner (i.e., the boundary of a 1D system) can be obtained from the phase difference between the mass terms \(\theta(x\rightarrow\pm\infty)\) on the adjacent edges (i.e., the bulk of a 1D system). From the perspective of the two time-reversal sectors described by Eq.(\ref{eq:12}), the total soliton charge can be considered as the contribution of two identical soliton charges. Each sub-Hamiltonian features exponentially localized domain-wall states at the edge intersections (i.e., the corner). Hence, we can expect pairs of identical soliton excitations localized around the corner and generate in-gap corner modes with extra charges. The intuitive views can be readily understood by changing the basis:
\begin{equation}
\begin{split}
\begin{pmatrix} \varphi_- \\ \varphi_+  \\ \vartheta_- \\ \vartheta_+ \end{pmatrix} = \frac{1}{2}\begin{pmatrix} 1 & 1 & 1 & -1 \\ 1 & 1 & -1 & 1  \\ 1 & -1 & 1 & 1 \\ -1 & 1 & 1 & 1 \end{pmatrix} \begin{pmatrix} \varphi_1 \\ \varphi_2  \\ \vartheta_1 \\ \vartheta_2 \end{pmatrix} .
\end{split}
\end{equation} 
In the new basis $\begin{pmatrix} \varphi_+ & \varphi_-  & \vartheta_+ & \vartheta_- \end{pmatrix}^T$, the Bosonized Hamiltonian becomes:
\begin{equation*}
\begin{split}
    H^{(B)} & = H^{(B)}_+ + H^{(B)}_-,
\end{split}
\end{equation*}
where $H^{(B)}_\pm$ are:
\begin{equation}\label{eq:26}
\begin{split}
H^{(B)}_{\pm} & = \frac{v}{2}[(\partial_x\varphi_\pm)^2 + (\partial_x\vartheta_\pm)^2 \mp \frac{2\delta v}{v}\partial_x\varphi_\pm\partial_x\vartheta_\pm] 
\\ & \quad \pm\frac{m}{\pi a}\cos(\sqrt{4\pi}\varphi_{\pm}+\theta) .
\end{split}
\end{equation}
According to Eq.~(\ref{eq:03}) and Eq.~(\ref{eq:05}), $H^{(B)}_\pm$ can be regarded as the Bosonized sub-Hamiltonian describing the massive Dirac fermions at the momenta $\pm k_0$ in Fig.~\ref{fig:02}. For a typical case, we have $\frac{\delta v}{v}\ll1$ such that terms with $\frac{\delta v}{v}$ can be neglected, and $H_\pm$ can be regarded as Gaussian models with the sine-Gordon terms \cite{book,book:17993,book:129250}. The coupling constants in Eq.~(\ref{eq:26}) correspond to the kink solution instead of the breather solution in a sine-Gordon problem, indicating both of $H^{(B)}_\pm$  produce the same amount of topological charges \cite{book:17993,book:129250}.

\par As will be further discussed below, instead of having topological in-gap corner zero modes, we will have topologically protected extra charge around the corner. Note that this charge is calculated with respect to the ground state \cite{PhysRevB.100.205126, PhysRevB.22.2099, RevModPhys.60.781}, which can be obtained through computing the \emph{charge fluctuation} in a finite-sized thin film and summing the value around the corner:
\begin{equation}
    \Delta\rho(x,y) = e\sum_{i=1}^{N_{occ.}} \braket{\psi_{i}(x,y) | \psi_{i}(x,y)} - eN_{occ.}^{(site)},
\end{equation}
where \(N_{occ.}\) is the number of the occupied bands and \(N_{occ.}^{(site)}\) is the number of the occupied bands per site. For our convenience, we will employ \emph{particle number fluctuation} $\Delta N(x,y)=\frac{\Delta\rho(x,y)}{e}$ in subsequent calculations.

\subsection{Symmetry properties}
Here, we discuss the symmetries of the systems under investigation. Throughout this article, we focus on time-reversal invariant settings, but the particle-hole and chiral symmetries can be broken by either the nonzero velocity difference \(\delta v\neq0\) or the coexistence of the two mass terms. In either case, the energy of the corner modes is not fixed at zero since the system is in class AII in one dimension, characterized by trivial topology \cite{RevModPhys.88.035005}. 
On the other hand, we can have topologically protected corner zero-modes when the particle-hole and chiral symmetries are restored if we have  \(\delta v=0\) with at most a single mass term, i.e., \(\theta=\frac{n\pi}{2}\,,n\in\mathbb{Z}\) on all edges. 
In the latter case, the topological invariant will be characterized by \(\mathbb{Z}_2\) of class DIII in one dimension \cite{PhysRevLett.111.056403,kainaris_santos_gutman_carr_2017}. For broader settings with both broken time-reversal and particle-hole symmetries, the systems fall into class AIII characterized by \(\mathbb{Z}\) in the presence of the chiral symmetry or class A without topological modes otherwise. Based on the above consideration, our system contains only topologically protected charges instead of states with fixed energies. 

\subsection{Fractionalization properties}
\par Our theory hosts an unusual relation between the charge and the spin degree of freedom originating from the helical nature of the edge states of a spin-Chern insulator, a distinct feature from the previous literature, in which the distinct Fermi velocities of different energy band branches lead to the charge-spin mixing \cite{10.1143/PTPS.108.265, PhysRevB.53.9572, PhysRevB.62.16900, PhysRevLett.84.4164, PhysRevB.100.195423}. In contrast, we demonstrate that the charge degree of freedom is intertwined with the spin degree of freedom in Appendix~\ref{section:A} due to the helical edge states, implying that they are not independent sectors. Instead, it is the charge-difference sector that remains independent from the charge sector \cite{PhysRevLett.108.196804}. The two sectors together thus describe the low-energy edge theory. Furthermore, the charge-difference sector denotes the deviation between the time-reversal subsystems. This is evident in Eq.~(\ref{eq:10}), where it interacts with the charge sector because of a nonzero difference in Fermi velocity. Still, solitons in our system demonstrate an unusual spin-charge relationship. They feature ``spin-charge separated'' excitations, meaning that they carry fractional charges without spin \cite{PhysRevLett.49.1455}, as depicted in Fig.~\ref{fig:02.4}(B). Due to time-reversal symmetry, there are no pure spinon excitations upon introducing many-body interactions. Instead, the many-body system hosts soliton and antisoliton excitations, carrying a fractional charge of $e\pm Q_s^0$. The two are related to each other by the time-reversal operator for half-integer spins \cite{PhysRevB.22.2099, PhysRevLett.101.086802, PhysRevLett.46.738, PhysRevLett.49.1455}, as indicated in Fig.~\ref{fig:02.4}(C) and Fig.~\ref{fig:02.4}(D).

\begin{figure}[ht]
  \centering
    \includegraphics[width=8.6cm, height=10cm]{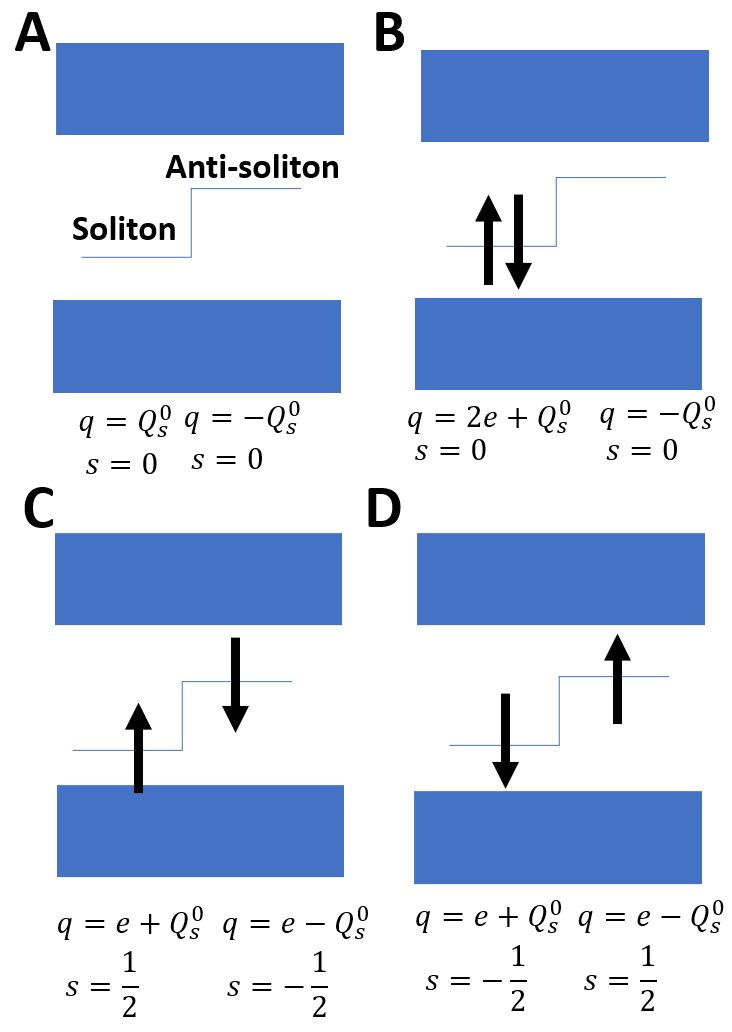}
  \caption{
  (A) In general, without zero-mode-protecting symmetries, soliton and anti-soliton excitations possess distinct energies and carry opposite charges $\pm Q_s^0$. (B) Due to time-reversal symmetry, a filled soliton excitation carries a net charge of $2e+Q_s^0$ without an associated spin. (C)(D) In a time-reversal symmetric many-body system composed of soliton and antisoliton excitations, it is possible to have excitations with a half-integer spin and nonvanishing fractional charges $e \pm Q_s^0$ for (anti)solitons.}
  \label{fig:02.4}
\end{figure}

\subsection{Effect of Interactions}
\par Since the interaction between electrons is expected to play a role in the 1D edges \cite{book:129250}, here we discuss its influence on the mass term through the RG analysis.  
To incorporate the interaction in Eq.~(\ref{eq:07}), we include the forward scattering in each time-reversal sector after Bosonizing it \cite{PhysRevLett.121.196801} and obtain two copies of helical Luttinger liquid with the corresponding interaction parameters $K_1,K_2$ characterizing the interaction strength in each sector:
\begin{equation}
    H_0^{(B)} = \frac{1}{2}\sum_{\mu=1}^2 v_\mu(K_\mu^{-1}(\partial_x\varphi_\mu)^2+K_\mu(\partial_x\vartheta_\mu)^2).
\end{equation}
In the basis of $\begin{pmatrix} \varphi_1 & \varphi_2  & \vartheta_1 & \vartheta_2 \end{pmatrix}^T$, Eq.~(\ref{eq:17}) can be rewritten as:
\begin{equation} \label{eq:sG-term}
    H_M^{(B)} = \frac{2m}{\pi a}\sin(\sqrt{\pi}(\vartheta_1-\vartheta_2))\sin(\theta+\sqrt{\pi}(\varphi_1+\varphi_2)) .
\end{equation}

If we define $\tilde{m}\equiv \frac{m}{\sqrt{v_1v_2}}a$, the RG flow equations can be obtained from the standard procedure \cite{PhysRevLett.121.196801, PhysRevB.97.045415, book:129250}:
\begin{align}
    \frac{dK_\mu}{dl} & = -\tilde{m}^2\frac{K_\mu^2-1}{2} ,  \label{eq:D5}
    \\ \frac{d\tilde{m}}{dl}   & = \big[ 2-\frac{1}{4}\sum_\mu (K_\mu+K_\mu^{-1}) \big] \tilde{m}, \label{eq:D6}
\end{align}
where we have set $a=a_0e^l$ with the short-distance cutoff $a_0$. Note that Eq.~(\ref{eq:D6}) can be written as:
\begin{align}\label{eq:D7}
    \frac{dm}{dl}  & = \big[ 1-\frac{1}{4}\sum_\mu (K_\mu+K_\mu^{-1}) \big] m. 
\end{align}

\par Since $K_\mu+K_\mu^{-1}\geq2$, the corresponding operator in Eq.~(\ref{eq:sG-term}) is at most marginally relevant (in the RG sense). It can be seen from Eqs.~(\ref{eq:D5}) and (\ref{eq:D6}) that $K_\mu$ and $m$ tend to flow toward the fixed points $(K_\mu,m) \to (1,m^*)$ or $ (K_\mu^*,0)$, with some renormalized values $m^*$ and $K_\mu^*$.
Intuitively, when $K_\mu(l=0)$ deviates from unity, $m$ decreases to zero faster than the evolution of $K_\mu$ to unity. Otherwise, the sine-Gordon term is marginal, and the system scales as its noninteracting version. To demonstrate this behavior, we analyze the RG flow around the fixed point $K_\mu=1$ by setting $K_\mu=1+\delta K_\mu$ with $\delta K_\mu\ll1$ and expanding Eqs.~(\ref{eq:D5}) and (\ref{eq:D6}) to the third-order of perturbations $p$, where $p\in\{\delta K_\mu, m\}$. 
This procedure gives 
\begin{align}
    \frac{d(\delta K_\mu)^2}{dl} & = -2\tilde{m}^2(\delta K_\mu)^2  , \label{eq:D8}
    \\ \frac{d\tilde{m}^2}{dl}   & = (2-\frac{1}{4}\sum_\mu(\delta K_\mu)^2)\tilde{m}^2. \label{eq:D9}
\end{align}
By setting $x_\pm=\tilde{m}^2\pm\frac{1}{8}\sum_\mu(\delta K_\mu)^2$, we can solve Eqs.~(\ref{eq:D8}) and (\ref{eq:D9}) exactly. The solutions after transforming $x_\pm$ back into $\hat{m}\equiv\frac{\tilde{m}}{a}=\frac{m}{\sqrt{v_1v_2}},r_K=\sqrt{\sum_\mu(\delta K_\mu)^2}$ are:
\begin{align}
    r_K(l) & = r_K(l=0) e^{-l} , \label{eq:D10}
    \\ \hat{m}^2(l)   & = \hat{m}^2(l=0)+\frac{e^{-4l}-1}{16}r_K^2(l=0). \label{eq:D11}
\end{align}
The trajectories of Eqs.~(\ref{eq:D10}) and (\ref{eq:D11}) with several sets of initial values $\{r_K(l=0)\}$ and $\{\hat{m}(l=0)\}$ are plotted in Fig.~\ref{fig:02.5}. According to Eqs.~(\ref{eq:D10}) and (\ref{eq:D11}), $m$ becomes irrelevant when $r_K(l=0)\geq4\hat{m}(l=0)$. On the other hand, the sine-Gordon term characterized by $m$ is marginal at large $l$ when $r_K(l=0)<4\hat{m}(l=0)$, which can be better captured if we expand Eqs.~(\ref{eq:D5}) and (\ref{eq:D7}) only to the first order of $\delta K_\mu$:
\begin{align}
    \frac{d\delta K_\mu}{dl} & = -\tilde{m}^2\delta K_\mu , \label{eq:D12}
    \\ \frac{dm}{dl} & = 0. \label{eq:D13}
\end{align}
Eqs.~(\ref{eq:D12}) and (\ref{eq:D13}) indicate that $\lvert\delta K_\mu\rvert$ decrease due to finite $m$ while $m$ remains unchanged during the RG flow. 

\par While the above RG analysis applies to infinite systems where the flow is governed by the initial values, in realistic systems, the RG flow can be stopped earlier due to finite size or finite temperature, thus preventing it from reaching the gapless fixed point with $m=0$ \cite{PhysRevLett.121.196801, PhysRevLett.123.036803}.
For instance, in a finite edge, the RG flow only evolves up to the scale $l\sim\text{ln}(L/a_0)$ with the edge length $L$, where the system has finite mass terms characterized by the renormalized $m^*$ value. 
As a result, time-reversal soliton pairs can still exist in a finite-sized thin film, even in an interacting system where $K_\mu$ deviates from unity.

\begin{figure}[ht]
  \centering
    \includegraphics[width=8.6cm, height=6cm]{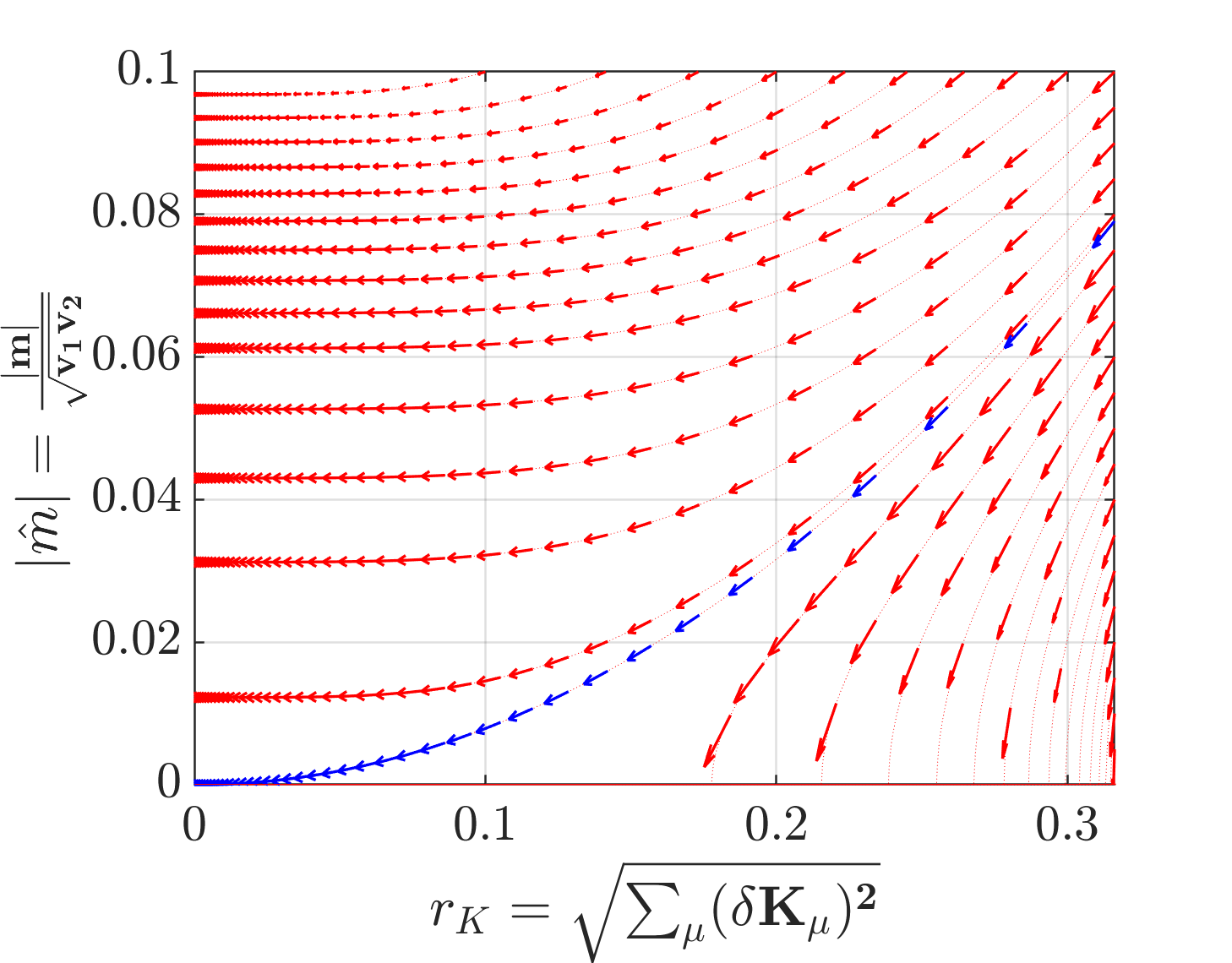}
  \caption{Solution of the RG flows in Eq.~(\ref{eq:D8}) and Eq.~(\ref{eq:D9}) described by Eq.~(\ref{eq:D10}) and Eq.~(\ref{eq:D11}). The blue trajectory indicates the critical points where $r_K(l=0)=4\hat{m}(l=0)$.}
  \label{fig:02.5}
\end{figure}

\section{Application to an extended Haldane model} \label{section:III}
\par In this section, we illustrate the previous results with a specific example from a time-reversal-symmetric extension of the Haldane model \cite{PhysRevB.104.L201110}. Before adding the $S_z$-mixing SOC,
gapless edge states emerge due to its spin-Chern number $C_s=2$.  
The presence of $S_z$-mixing SOC opens small gaps in the edge band structure. 
The Hamiltonian of this model is:
\begin{eqnarray}\label{eq:22_new} 
\begin{split}
    H = & \sum_{\sigma}[d_1(\vec{k})a_{\sigma}^{\dagger}b_{\sigma} - id_2(\vec{k})a_{\sigma}^{\dagger}b_{\sigma} + H.c. \\
       & + (d_3(\vec{k})\text{sign}({\sigma}) + m)(a_{\sigma}^{\dagger}a_{\sigma} - b_{\sigma}^{\dagger}b_{\sigma})]\\
       & + i\lambda_R\sum_{j=1}^{3}(\vec{c}_j\times\vec{S})_{z}^{\sigma\sigma'}e^{i\vec{k}\cdot\vec{c}_j}a_{\sigma}^{\dagger}b_{\sigma'} + H.c.,
\end{split}
\end{eqnarray} 
where $\sigma$ is the spin index with $a_{\sigma}^{\dagger},b_{\sigma}^{\dagger}$ be the creation operator of a spin $\sigma$ electron on different sublattice. The $d_1(\vec{k}),d_2(\vec{k}),d_3(\vec{k})$ are defined to be $d_1=\sum_{j=1}^{3}[t_1\cos(\vec{k})\cdot\vec{a}_j+t_3\cos(\vec{k}\cdot\vec{c}_j)]$, $d_2=\sum_{j=1}^{3}[-t_1\sin(\vec{k}\cdot\vec{a}_j)-t_3\sin(\vec{k}\cdot\vec{c}_j)]$, and $d_3=\sum_{j=1}^{6}t_2(-1)^j\sin(\vec{k}\cdot\vec{b}_j)$, where $\vec{a}_j,\vec{b}_j,\vec{c}_j$ are the vectors connecting a site to its first-, second-, and third-nearest neighbors on a honeycomb lattice, respectively. In the following calculations, we set $t_1=1$, $m=0.1$, $\lambda_R=0.3$, and $t_2=t_3=0.6$. With these parameters, the model has spin-Chern number $C_s=2$ and $S_z$-mixing SOC gaps of edge states on its armchair edges (see Fig.~\ref{fig:03}). In both cases, the centers of the gapped Dirac cones are shifted in both $k$-space and energy, indicating a $k$-dependent mass term in the low-energy theory. Still, as a zeroth order approximation in $k$ for the mass term, our theory still well explained the phenomena of the low-energy theories in these cases. 

\begin{figure}[ht]
  \centering
    \includegraphics[width=8.6cm, height=5cm]{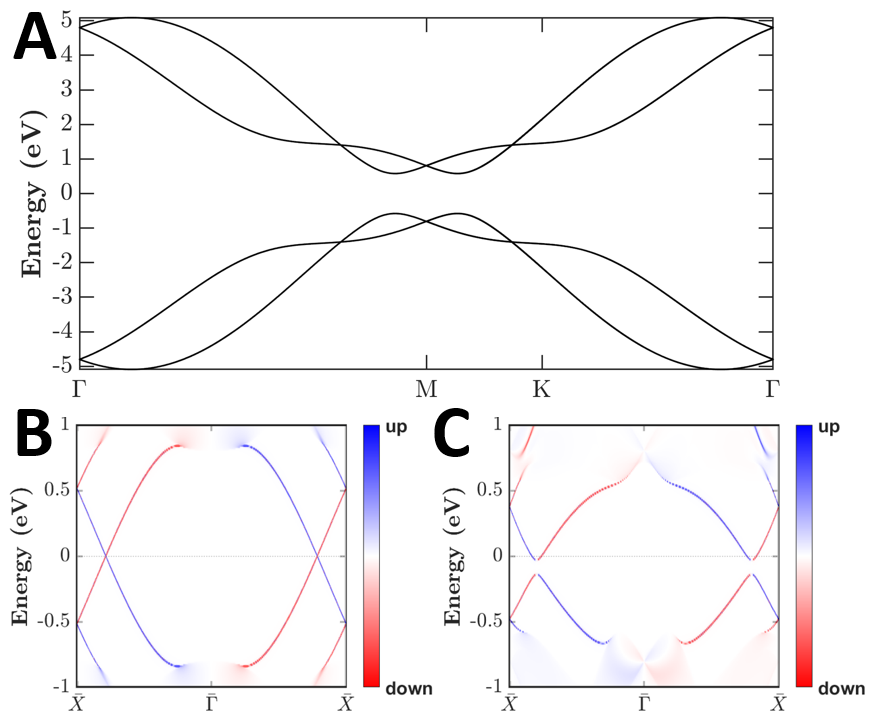}
  \caption{(A.) The bulk band structure of the Hamiltonian in Eq.~\ref{eq:22_new}. (B.) The spin-resolved spectrum of the gapless edge states (C.) The spin-resolved spectrum of the edge states that are gapped by the Rashba SOC. The Dirac cones shift toward the Brillouin zone boundary in $k$-space and to lower energy due to the Rashba SOC.}
  \label{fig:03}
\end{figure}

\begin{figure}[ht]
  \centering
    \includegraphics[width=8.6cm, height=4cm]{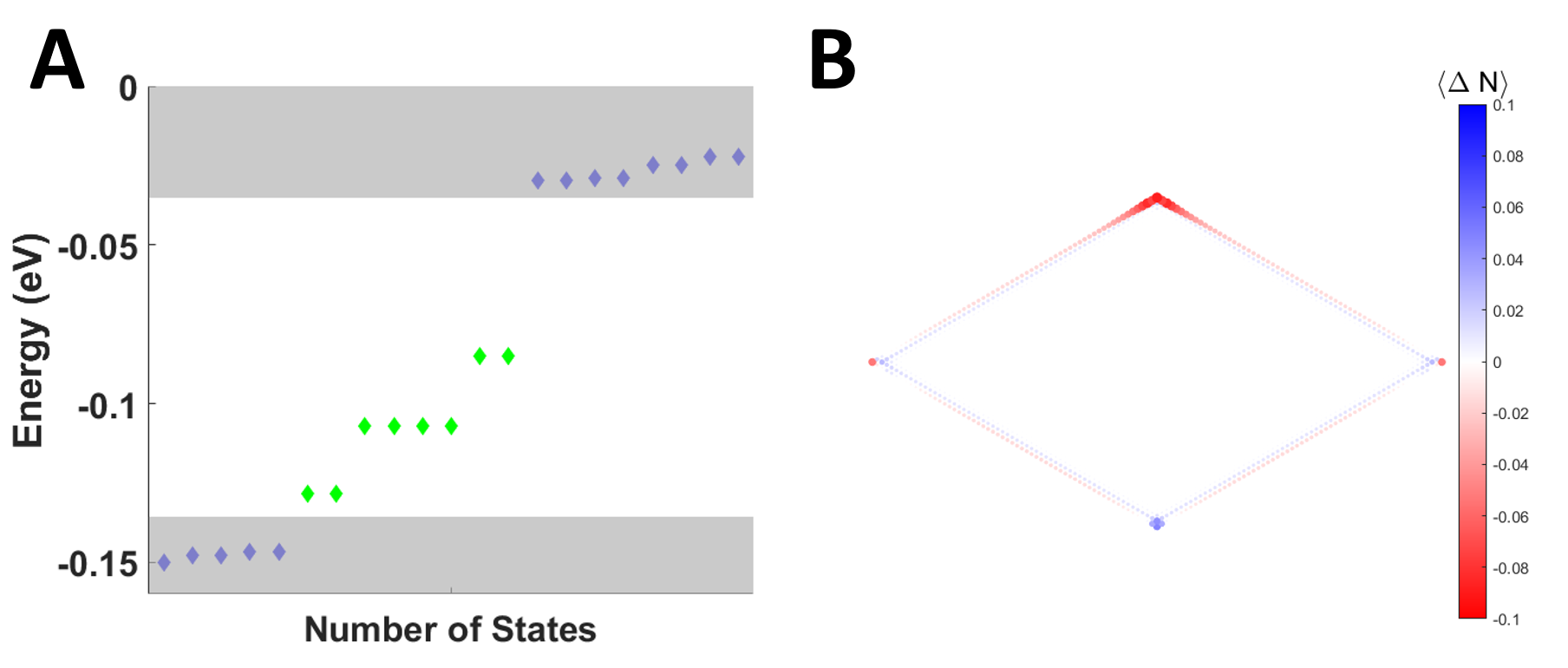}
  \caption{(A.) The $60\times60$ thin film spectrum, in which the in-gap corner modes are marked in green (B.) The particle number fluctuation $\Delta N$ with respect to the ground state in a $60\times60$ thin film. }
  \label{fig:04}
\end{figure}

\par After the Rashba SOC gaps the edge states, the in-gap corner modes emerge in a finite-sized armchair-edged thin film (see Fig.~\ref{fig:04}(A.)), and the extra charges accumulate around the corner of the thin film (see Fig.~\ref{fig:04}(B.)). 

\subsection{Edge Hamiltonian for an Arbitrary Edge}
\par In this section, we derive the effective edge Hamiltonian on an arbitrary edge to demonstrate how Rashba SOC induces the edge mass-kink in the model described by Eq.~(\ref{eq:22_new}). The Bloch Hamiltonian of Eq.~(\ref{eq:22_new}) is:
\begin{equation}\label{eq:23_new}
\begin{split}
    H & = [t_1(1+\cos(\frac{k_x+\sqrt{3}k_y}{2})+\cos(\frac{-k_x+\sqrt{3}k_y}{2}))]\sigma_0\tau_1
    \\ & \quad + [t_1(1+\sin(\frac{k_x+\sqrt{3}k_y}{2})+\sin(\frac{-k_x+\sqrt{3}k_y}{2}))]\sigma_0\tau_2
    \\ & \quad +[t_3(2\cos(k_x)+\cos(\sqrt{3}k_y))]\sigma_0\tau_1 +[t_3\sin(\sqrt{3}k_y)]\sigma_0\tau_2
    \\ & \quad + 2t_2(\sin(\frac{k_x+\sqrt{3}k_y}{2})-\sin(\frac{-k_x+\sqrt{3}k_y}{2}))\sigma_3\tau_3
    \\ & \quad + m\sigma_0\tau_3 - 2t_2\sin(k_x)\sigma_3\tau_3
    \\ & \quad + 2\lambda_R\sin(\sqrt{3}k_y)\sigma_1\tau_1 - 2\sqrt{3}\lambda_R\sin(k_x)\sigma_2\tau_1
    \\ & \quad + 4\lambda_R\sin(\frac{k_x+\sqrt{3}k_y}{2})\sin(\frac{-k_x+\sqrt{3}k_y}{2})\sigma_1\tau_2,
\end{split}
\end{equation}
where 
$\sigma$ ($\tau$) are the Pauli matrices in the spin (sublattice) degree of freedom. 
According to its band structure (see Fig.~\ref{fig:03}), the low-energy Hamiltonian is obtained through expanding $(k_x,k_y)$ around point $M$ in its Brillouin Zone:
\begin{equation}\label{eq:24_new}
\begin{split}
    H & = (3t_3-t_1)\sigma_0\tau_1 + (t_1+\sqrt{3}(t_3-t_1)k_y)\sigma_0\tau_2 + m\sigma_0\tau_3 
    \\ & \quad + 2\sqrt{3}\lambda_R(k_y\sigma_1\tau_1 - k_x\sigma_2\tau_1) .
\end{split}
\end{equation}

To understand the low-energy Hamiltonian on arbitrary edges, we describe the coordinates in the low-energy Hamiltonian with a new basis set \cite{PhysRevB.99.155102, Mu2022}:
\begin{equation}
\begin{split}
\begin{pmatrix} k_x \\ k_y \end{pmatrix} = \begin{pmatrix} \sin(\xi) & \cos(\xi) \\ -\cos(\xi) & \sin(\xi) \end{pmatrix} \begin{pmatrix} k_1 \\ k_2 \end{pmatrix},
\end{split}
\end{equation}
and get: 
\begin{equation}
\begin{split}
    H & = (3t_3-t_1)\sigma_0\tau_1 + m\sigma_0\tau_3
    \\ & \quad + (t_1+\sqrt{3}(t_3-t_1)(i\cos(\xi)\partial_x\pm k_0\sin(\xi))\sigma_0\tau_2 
    \\ & \quad + 2\sqrt{3}\lambda_R(i\cos(\xi)\partial_x\pm k_0\sin(\xi))\sigma_1\tau_1 
    \\ & \quad - 2\sqrt{3}\lambda_R(-i\sin(\xi)\partial_x\pm k_0\cos(\xi))\sigma_2\tau_1 .
\end{split}
\end{equation}
By using the ansatz $e^{\eta x}\phi$, we can obtain the chiral edge states before adding the Rashba term:
\begin{equation}\label{eq:27}
\begin{split}
    \ket{R_{\pm}} & = N_{\pm}\begin{pmatrix} \frac{A_\pm(\xi)e^{i\phi_\pm(\xi)}}{m} & -1 & 0 & 0 \end{pmatrix}^T , 
    \\ \ket{L_{\pm}} & = N_{\pm}\begin{pmatrix} 0 & 0 &  \frac{A_\pm(\xi)e^{i\phi_\pm(\xi)}}{m}  & -1 \end{pmatrix}^T,
\end{split}
\end{equation}
and the Rashba term:
\begin{equation}\label{eq:28}
    i2\sqrt{3}\lambda_R[(\eta\cos(\xi)\mp ik_0\sin(\xi))\sigma_1 + (\eta\sin(\xi)+ik_0\cos(\xi))\sigma_2]\tau_1 ,
\end{equation}
where:
\begin{equation*}
\begin{split}
A_\pm(\xi) & = \|(3t_3-t_1)-i(t_1+\sqrt{3}(t_3-t_1)k_\pm)\| ,
\\ \phi_\pm(\xi) & = Arg((3t_3-t_1)-i(t_1+\sqrt{3}(t_3-t_1)k_\pm))
\\ k_\pm & = i\eta \cos(\xi)\pm k_0\sin(\xi),
\end{split}
\end{equation*}
and $N_{\pm}$ is the normalization factor. Since the edge states are very localized, we assume $\eta\gg 1$ such that Eq.~\ref{eq:28} can be approximated by:
\begin{equation}
    i2\sqrt{3}\lambda_R\eta[\cos(\xi)\sigma_1 + \sin(\xi)\sigma_2]\tau_1 ,
\end{equation}
with $\phi_\pm(\xi)\cong0$. The matrix elements of the mass term $m_\pm$, for example, are $\bra{R_\pm}\Delta H\ket{L_\mp}$, where $\Delta H$ is the Hamiltonian in Eq.~(\ref{eq:28}). Therefore, the phase of $m_\pm$ is $e^{-i(\xi+\frac{\pi}{2})}$, which has different values on different edges, indicating the Rashba SOC induces mass-kinks on the edges.

\section{Summary and Conclusions} \label{section:IV}
\par We derive the mass-kink mechanism in a 1D system with two bulk Dirac cones related by time-reversal symmetry using Bosonization techniques. Different bulk Dirac cones generate two solitons on the domain wall by imposing gap-opening time-reversal-symmetric perturbations. The value of the extra charge on the domain wall is thus twice the value of the phase difference in the mass term on the Dirac cones. The theory is applied to the edge of a spin-Chern insulator with $C_s=2$, in which the corner modes in a finite-sized thin film result from solitons generated through the mass-kink mechanism induced by $S_z$-mixing SOC. We consider the model proposed in \cite{PhysRevB.104.L201110} as an example. By investigating the phase evolution of the corner modes along the edge of a finite-sized thin film, we show that $S_z$-mixing SOC opens the gapless edge states and, indeed, induces solitons through the mass-kink mechanism. The phase evolutions also support the viewpoint that the pairwise in-gap corner modes are generated by different edge Dirac cones. 

With these results in mind, we construct an HO topological phase with protected corner charges from a spin-Chern insulator with $C_s=2$.
Since the key ingredient, the complex mass term with the edge-dependent phase is generically present here (induced, for example, by the $S_z$-mixing SOC), our analysis provides strong evidence that the results hold for general systems with spin-Chern number $C_s=2$. 
Our formalism can be extended straightforwardly to treat systems with a higher number of time-reversal-related Dirac cones, 
indicating that time-reversal soliton pairs with protected corner charges can appear on the edges of spin-Chern insulators with even spin-Chern numbers $C_s=2n, n>1$ more generally. Similar arguments would also apply to $\hat{w}$-Chern insulators with a non-zero even $\hat{w}$-Chern number and $\hat{w}$-symmetry breaking perturbations \cite{PhysRevB.80.125327,shulman2010robust,wang2023feature}.

\section*{Acknowledgement}
\par The authors thank Yun-Chung Chen and Benjamin J. Wieder for useful discussions. The work at Northeastern University was supported by the National Science Foundation through NSF-ExpandQISE award \#2329067 and it benefited from Northeastern University’s Advanced Scientific Computation Center and the Discovery Cluster. H.L. acknowledges the support by the National Science and Technology Council (NSTC) in Taiwan under grant number MOST 111-2112-M-001-057-MY3. C.H. acknowledges the financial support from the National Science and Technology Council (NSTC), Taiwan through NSTC-112-2112-M-001-025-MY3.


\appendix
\section{Bosonization conventions}\label{section:A}
\par By using the notation in Appendix A of Ref.\cite{PhysRevB.92.035139}, the Bosonic representations of the right and left moving modes are:
\begin{equation}\label{eq:01}
\begin{split}
    R_\sigma(x) & = \frac{\kappa_\sigma}{\sqrt{2\pi a}}e^{i\sqrt{4\pi}\phi_{\sigma}^{R}(x)} , \\
     L_\sigma(x) & = \frac{\kappa_\sigma}{\sqrt{2\pi a}}e^{-i\sqrt{4\pi}\phi_{\sigma}^{L}(x)},
\end{split}
\end{equation}
where \(a\rightarrow 0\) is the Bosonic UV cutoff, \(\sigma=1,2\) denotes the time-reversal sector and \(\kappa_\sigma\) is the Klein factors satisfying \( \{ \kappa_\sigma,\kappa_{\sigma'}\}=2\delta_{\sigma\sigma'}\) and \(\kappa_\sigma^2=1, \kappa_1\kappa_2=-\kappa_2\kappa_1=i\). Since all right moving modes carry spin \(\uparrow\) while all left moving modes carry spin \(\downarrow\), we suppress the index for spin in Eq.(\ref{eq:01}).

To have correct anti-commutation relation \(\{R_\sigma, L_\sigma\}=0\) and \(\{R_\sigma(x), R_\sigma(y)\}=0\), the Boson field should follow the commutation relation:
\begin{equation}\label{eq:02}
\begin{split}
    & [\phi_\sigma^R(x),\phi_{\sigma'}^L(x)]=\frac{i}{4}\delta_{\sigma\sigma'}  ,
    \\ & [\phi_\sigma^\eta(x),\phi_{\sigma'}^{\eta'}(y)]=\frac{i}{4}\eta\delta_{\eta,\eta'}\delta_{\sigma\sigma'}\text{sign}(x-y) .
\end{split}
\end{equation}

The conjugated variables are defined as:
\begin{eqnarray}\label{eq:03}
\begin{split}
    & \varphi_\sigma \equiv \phi^R_\sigma + \phi^L_\sigma , \\
    & \vartheta_\sigma \equiv \phi^L_\sigma - \phi^R_\sigma .
\end{split}
\end{eqnarray}

Then, we have:
\begin{align}\label{eq:04}
    R^\dagger_\sigma R_\sigma + L^\dagger_\sigma L_\sigma = \frac{1}{\sqrt{\pi}}\partial_x\varphi_\sigma .
\end{align}

Further, we can define the charge degree of freedom and the charge-difference degree of freedom:
\begin{equation}\label{eq:05}
\begin{split}
    \varphi_c & \equiv \frac{1}{\sqrt{2}}(\varphi_1 + \varphi_2) ,
    \\ \vartheta_c & \equiv \frac{1}{\sqrt{2}}(\vartheta_1 + \vartheta_2) ,
    \\ \varphi_d & \equiv \frac{1}{\sqrt{2}}(\varphi_1 - \varphi_2) ,
    \\ \vartheta_d & \equiv \frac{1}{\sqrt{2}}(\vartheta_1 - \vartheta_2).
\end{split}
\end{equation}
Note that, by using Eq.~(\ref{eq:03}) and Eq.~(\ref{eq:04}), we can get $\partial_x\vartheta_c \propto \rho_s$, where $\rho_s$ is the spin density. The spin degree of freedom and the charge degree of freedom are related through these conjugated variables. Similarly, the spin and charge density differences are also related in the same way.

\section{Basic Properties of Bosonic Fields} \label{section:B}
\begin{theorem}\label{thm:01}
$[\varphi_\sigma(x),\partial_y\vartheta_\sigma(y)]=i\pi\delta(x-y)$\,,\(\sigma = 1,2\)
\end{theorem}
\begin{proof}
In our model, the fermionic annihilation operator \(\psi\) can be written as:
\begin{eqnarray}
\begin{split}
    \psi = &\sum_{\sigma}(\sum_{k>0}\frac{e^{ikx}}{\sqrt{L}}c_{k,\sigma} + \sum_{k<0}\frac{e^{ikx}}{\sqrt{L}}c_{k,\sigma})  \\
     = &\sum_{\sigma}(\sum_{k=-k_0}^{\infty}\frac{e^{i(k+k_0)x}}{\sqrt{L}}C_{k+k_0,\sigma}\\
      &+  \sum_{k=-\infty}^{k_0}\frac{e^{i(k-k_0)x}}{\sqrt{L}}c_{k-k_0,\sigma})\\
    \equiv& e^{ik_0x}(R_1 + L_2) + e^{-ik_0x}(R_2 + L_1),
\end{split}
\end{eqnarray}
where \(L\rightarrow\infty\) is the size of the system and \(\psi(x+L)=\psi(x)\) is satisfied. The right(left) moving modes are defined as:

\begin{eqnarray}
\begin{split}
    & R_1 = \sum_k \frac{e^{ikx}}{\sqrt{L}}C_{k}^{R_1}\,,C_{k}^{R_1} \equiv C_{k+k_0,\uparrow} 
    \\ & L_1 = \sum_k \frac{e^{-ikx}}{\sqrt{L}}C_{k}^{L_1}\,,C_{k}^{L_1} \equiv C_{-(k+k_0),\downarrow}
    \\ & R_2 = \sum_k \frac{e^{ikx}}{\sqrt{L}}C_{k}^{R_2}\,,C_{k}^{R_2} \equiv C_{k-k_0,\uparrow}
    \\ & L_2 = \sum_k \frac{e^{-ikx}}{\sqrt{L}}C_{k}^{L_2}\,,C_{k}^{L_2} \equiv C_{-k+k_0,\downarrow}
\end{split}
\end{eqnarray}
In this definition, the right(left) moving modes are filled when \(k<0\) in each summation.
\par Then, following \cite{miranda_2003}, each right(left) moving mode is the coherent state of the corresponding Boson field:
\begin{eqnarray}
\begin{split}
    [b_q^{R_\mu},R_\mu] & = \sqrt{\frac{2\pi}{L^2|q|}}[\sum_kc_{k-q}^{(R_\mu)\dagger}c_k^{R_\mu},\sum_{k'}e^{ik'x}c_{k'}^{R_\mu}]
    \\ & = -\sqrt{\frac{2\pi}{|q|}}\frac{1}{L}\sum_{k,k'}\delta_{k-q,k'}e^{ik'x}c_k^{R_\mu}
    \\ &  = -\sqrt{\frac{2\pi}{|q|L}}e^{-iqx}R_\mu ,
\end{split}
\end{eqnarray}
\begin{equation}
\begin{split}
    [b_q^{L_\mu},L_\mu] & = \sqrt{\frac{2\pi}{L^2|q|}}[\sum_kc_{k-q}^{(L_\mu)\dagger}c_k^{L_\mu},\sum_{k'}e^{-ik'x}c_{k'}^{L_\mu}]
    \\ & = -\sqrt{\frac{2\pi}{|q|}}\frac{1}{L}\sum_{k,k'}\delta_{k-q,k'}e^{-ik'x}c_k^{L_\mu}
     \\ & = -\sqrt{\frac{2\pi}{|q|L}}e^{iqx}L_\mu,
\end{split}
\end{equation}
where \(b_q^{R_\mu(L_\mu)}\) is the Bosonic annihilation operator for the density fluctuation of each right(left) moving mode. Since \(R_\mu,L_\mu\) satisfy the commutation relation for right(left) moving modes defined in \cite{miranda_2003} for each \(\mu\). Therefore, as derived for the right(left) moving modes in \cite{miranda_2003}, \([\varphi_\mu(x),\partial_y\vartheta_\mu(y)]=i\pi\delta(x-y)\) for \(\mu=1,2\) denoting the time-reversal sector. 
\end{proof}
Theorem (\ref{thm:01}) shows that \(\varphi_\mu\) and \(\vartheta_\mu\) are indeed canonical conjugate variables. By using Theorem (\ref{thm:01}), the commutation relation \([\varphi_{c(d)}(x),\partial_y\vartheta_{c(d)}(y)]=i\pi\delta(x-y)\) can also be proofed, which means that the Boson field and its dual field in charge sector and charge-difference sector are canonical conjugate variables, too. Besides, this also provides a foundation for Eq.(\ref{eq:04}) by using the corresponding derivation in \cite{miranda_2003}.
\\
\section{General Mass Term} \label{section:C}
\par In general, the gap-opening perturbation is not necessary to be time-reversal symmetric. Therefore, the phases of the mass-term of different Dirac cones can differ in an arbitrary phase $\alpha$. Then, Eq.~(\ref{eq:15}) becomes:

\begin{equation}
    H_M = me^{-i\theta}R^\dagger_1L_2 +  me^{-i(\theta+\alpha)}R^\dagger_2L_1 + H.c.,
\end{equation}
which can be further Bosonized into:
\begin{equation}
\begin{split}
    H_M^{(B)} & = \frac{m}{\pi a}(\cos(\sqrt{2\pi}(\varphi_c-\vartheta_d)+\theta) 
    \\ & \quad + \cos(\sqrt{2\pi}(\varphi_c+\vartheta_d)+\theta+\alpha)) .
\end{split}
\end{equation}
It has minima at,
\begin{align}
    & \vartheta_d=-\frac{\alpha}{\sqrt{8\pi}} + (m-n)\sqrt{\frac{\pi}{8}} , \\ 
    & \varphi_c=-\frac{\theta}{\sqrt{2\pi}}-\frac{\alpha}{\sqrt{8\pi}} + (m+n)\sqrt{\frac{\pi}{8}},
\end{align}
for some $m,n\in\mathbb{Z}$ such that $m+n\in2\mathbb{Z}$. Therefore, the approximation made in Eq.~(\ref{eq:25}) is still valid. That is, there are always extra charges around the corner as long as a $k$-independent perturbation opens the gapless surface Dirac cones that generate different phases of the mass term on adjacent edges. And the amount of the extra charge can be calculated through Eq.~(\ref{eq:25.5}).


\bibliographystyle{apsrev4-2}
\bibliography{ref}

\end{document}